\documentclass[12pt]{article}
\font\smallit=cmti10

\usepackage{amssymb,latexsym}

\makeatletter

\renewcommand\section{\@startsection {section}{1}{\z@}
{-30pt \@plus -1ex \@minus -.2ex}
{2.3ex \@plus.2ex}
{\normalfont\normalsize\bfseries}}

\renewcommand\subsection{\@startsection{subsection}{2}{\z@}
{-3.25ex\@plus -1ex \@minus -.2ex}
{1.5ex \@plus .2ex}
{\normalfont\normalsize\bfseries}}

\renewcommand{\@seccntformat}[1]{\csname the#1\endcsname. }

\usepackage{amsmath}
\newtheorem{theorem}{Theorem}
\newtheorem{lemma}{Lemma}
\newtheorem{result}{Result}

\newenvironment{proof}
{\vspace{1ex}\noindent{\it Proof.}\hspace{0.5em}} 
{\hfill $\Box$}

\makeatother

\begin{document}

\begin{center}
\uppercase{\bf Extension of Stanley's Theorem for Partitions}
\vskip 20pt
{\bf Manosij Ghosh Dastidar}\\
{\smallit South Point High School, Kolkata, India}\\
{\tt gdmanosij@gmail.com}\\ 
\vskip 10pt
{\bf Sourav Sen Gupta}\\
{\smallit Indian Statistical Institute, Kolkata, India}\\
{\tt sg.sourav@gmail.com}\\ 
\end{center}

\vskip 30pt

\centerline{\bf Abstract}

\noindent
In this paper we present an extension of Stanley's theorem related to partitions of positive integers. Stanley's theorem states a relation between ``the sum of the numbers of distinct members in the partitions of a positive integer $n$'' and ``the total number of 1's that occur in the partitions of $n$''. Our generalization states a similar relation between ``the sum of the numbers of distinct members in the partitions of $n$'' and the total number of 2's or 3's or any general $k$ that occur in the partitions of $n$ and the subsequent integers. We also apply this result to obtain an array of interesting corollaries, including alternate proofs and analogues of some of the very well-known results in the theory of partitions. We extend Ramanujan's results on congruence behavior of the `number of partition' function $p(n)$ to get analogous results for the `number of occurrences of an element $k$ in partitions of $n$'. Moreover, we present an alternate proof of Ramanujan's results in this paper.

\thispagestyle{empty}
\baselineskip=15pt
\vskip 30pt 

\section{Introduction}
Partitioning a positive integer $n$ as sum of certain positive integers is a well known problem in the domain of number theory~\cite{wikin}. One of the very well referred results in this area is the one presented by Stanley~\cite{stanley1} which states the following.
\begin{center}
{\em ``The total number of 1's that occur among all unordered partitions of a positive integer is equal to the sum of the numbers of distinct members of those partitions.''}
\end{center}
One direction of generalizing Stanley's theorem is the Elder's theorem~\cite{elder1} which states that
\begin{center}
{\em ``total number of occurrences of an integer $k$ among all unordered partitions of $n$ is equal to the number of occasions that a part occurs $k$ or more times in a partition.''}
\end{center} 
In case of Elder's statement, a partition which contains $r$ parts that each occur $k$ or more times contributes $r$ to the sum in question.

\subsection{Our Contribution}
In this paper we generalize Stanley's theorem in a different direction. Let $S(n)$ be the number of distinct members present in the unordered partitions of $n$ and $Q_k(n)$ be the number of occurrences of $k$ in all the unordered partitions of $n$. Then Stanley's theorem says that given any positive integer $n$, the set of unordered partitions of $n$ satisfy $S(n) = Q_1(n)$. Let us present an example for $n=4$ corresponding to Stanley's theorem in Table~\ref{tab1}.

\begin{table}[htb]
\centering
\begin{tabular}{|c|c|c|}
\hline
\multicolumn{3}{|c|}{Stanley's theorem for $n=4$}\\
\hline
\hline
Partitions & Number of Distinct Members  & Number of 1's\\
\hline
4 & 1 & 0 \\
\hline
3+1 & 2 & 1 \\
\hline
2+2 & 1 & 0 \\
\hline
2+1+1 & 2 & 2 \\
\hline
1+1+1+1 & 1 & 4 \\
\hline
\hline
Total & {\bf 7} & {\bf 7} \\
\hline
\end{tabular}
\caption{Example of Stanley's theorem}
\label{tab1}
\end{table}

\noindent We extend this to prove that given any two positive integers $n$ and $k$, $S(n) = \sum_{i=0}^{k-1} Q_k(n+i)$.

\subsection{Idea for Extension}
The main idea behind the extension is the observation that Stanley's theorem can be extended to a relation between the distinct members of the partitions of $n$ and the number of occurrences of any positive integer $k$ in the partition table of $n$ and subsequent integers. Let us first present a few examples to illustrate our motivation.
\begin{description}
\item [$k=1$:] The sum of the numbers of distinct members of all the unordered 
partitions of a positive integer $n$ is equal to the total number of 1's that 
occur among the partitions of $n$. [This is Stanley's theorem]
\item [$k=2$:] The sum of the numbers of distinct members of all the unordered 
partitions of a positive integer $n$ is equal to the total number of 2's that 
occur among the partitions of $n$ and $n+1$.
\item [$k=3$:] The sum of the numbers of distinct members of all the unordered 
partitions of a positive integer $n$ is equal to the total number of 3's that 
occur among the partitions of $n$, $n+1$ and $n+2$.
\end{description}
One may note that Stanley's theorem is the special case with $k=1$, as mentioned above. The claim of this paper is that the statements made above are correct and can be extended to any value of $k$. The general result is formally presented in the next section.

\section{Extension of Stanley's Theorem}
Before we discuss the results, let us set some notational conventions. We mention $S(n)$ and $Q_k(n)$ once again for the sake of completion.
\begin{itemize}
\item $P(n)$: Number of unordered partitions of a positive integer $n$.
\item $S(n)$: Number of {\em distinct members} present in the unordered partitions 
of $n$.
\item $R_k(n)$: Number of partitions of $n$ containing the positive integer $k$.
\item $Q_k(n)$: Number of occurrences of $k$ in all the unordered partitions of $n$.
\end{itemize}
Based on the notation stated above, Stanley's theorem can be expressed as follows.

\begin{theorem}[Stanley's Theorem]
\label{thm1}
Given any positive integer $n$, the set of unordered partitions of $n$ 
satisfy $$S(n) \: = \: Q_1(n).$$
\end{theorem}

The extension to Stanley's theorem is as stated in Theorem~\ref{thm2}. But before that, we will prove the following two technical results which are used in the proof of Theorem~\ref{thm2}.

\subsection{Preliminary Technical Results}
\begin{lemma}
\label{lem1}
Given any positive integer $n$ and any positive integer $k$,
$$Q_k(n+k) = Q_k(n) + R_k(n+k).$$
\end{lemma}
\begin{proof}
Let us denote the set of partitions of $n$ containing $k$ as 
$$ X = \{ A_1, A_2, \ldots, A_r \},$$ 
where $r = R_k(n)$ is the number of such partitions, and the partitions are 
written in additive form (e.g., $5 = 2+2+1$).

The set defined above gives us $n = A_1 = A_2 = \cdots = A_r$, and thus 
$n + k = A_1 + k = A_2 + k = \cdots = A_r + k$. Now all the partitions in this 
collection $$ Y = \{A_1 + k, A_2 + k, \ldots, A_r + k\}$$ 
are partitions of $n+k$ containing $k$, but this is not an exhaustive list. 

There are additional partitions of $n+k$ containing $k$. One may note that 
the list $Y$ contains $r = R_k(n)$ partitions, whereas the number of partitions
of $n+k$ containing $k$ is $R_k(n+k)$. Thus, there are $R_k(n+k) - R_k(n)$ 
additional partitions of $n+k$ containing $k$. In this direction, we claim the 
following.

{\flushleft \sf Claim:} In each of the additional $R_k(n+k) - R_k(n)$ 
partitions of $n+k$ containing $k$, the member $k$ occurs only once.

{\flushleft \sf Proof:} Suppose the claim is false, that is, there exists a 
partition $B_1 \not\in Y$ of $n+k$ that contains more than one copy of $k$. 
Then $B_1 - k$ must contain at least one copy of $k$. Note that $n+k = B_1$, 
and thus $n = B_1 - k$, which implies that $B_1 - k$ is a partition of $n$ 
containing $k$. If so, then one must have $B_1 - k \in X$ and therefore 
$B_1 \in Y$. But this poses a contradiction, and hence the claim is true.

{\flushleft In view of our previous discussion and the claim above, we obtain 
the following}
\begin{itemize}
\item Each partition of $n+k$ in $Y$ contains $k$, and there are $r = R_k(n)$ 
partitions in $Y$.
\item There exist $R_k(n+k) - R_k(n)$ additional partitions of $n+k$ which 
contain $k$.
\item Each of these additional partitions contain exactly one copy of $k$.
\end{itemize}
Moreover, one may note that the total number of $k$'s in the partitions 
belonging to $X$ is $Q_k(n)$, as we selected all partitions of $n$ containing 
$k$. So, total number of $k$'s in the partitions belonging to $Y$ is 
$Q_k(n) + R_k(n)$, as we have added a $k$ to each partition from $X$. Thus, 
the total number of $k$'s occurring in the partitions of $n+k$ is given by
$$Q_k(n+k) = \left[Q_k(n) + R_k(n)\right] + \left[R_k(n+k) - R_k(n)\right] 
= Q_k(n) + R_k(n+k).$$
The analysis holds for all $n, k > 0$, and hence the result.
\end{proof}

\begin{lemma}
\label{lem2}
Given any positive integer $n$ and any positive integer $k$,
$$P(n) = R_k(n+k).$$
\end{lemma}
\begin{proof}
Let us denote the set of all partitions of $n$ as 
$$ X = \{ A_1, A_2, \ldots, A_p \},$$ 
where $p = P(n)$ is the number of all partitions, and the partitions are 
written in additive form (e.g., $5 = 2+2+1$).

The set defined above gives us $n = A_1 = A_2 = \cdots = A_p$, and thus 
$n + k = A_1 + k = A_2 + k = \cdots = A_p + k$. Now all the partitions in this 
collection $$Y = \{A_1 + k, A_2 + k, \ldots, A_p + k\}$$ 
are partitions of $n+k$ containing $k$. In this direction, we claim the 
following.

{\flushleft \sf Claim:} There are no other partitions of $n+k$, apart from 
those in $Y$, that contain $k$.

{\flushleft \sf Proof:} Suppose the claim is false, that is, there exists a 
partition $B_1 \not\in Y$ of $n+k$ that contains $k$. Note that $n+k = B_1$, 
and thus $n = B_1 - k$, which implies that $B_1 - k$ is a partition of $n$. 
If so, then one must have $B_1 - k \in X$ and therefore $B_1 \in Y$. But this 
poses a contradiction, and hence the claim is true.

{\flushleft In view of our previous discussion and the claim above, we obtain 
the following}
\begin{itemize}
\item Each partition of $n+k$ in $Y$ contains $k$, and there are $p = P(n)$ 
partitions in $Y$.
\item There exist no additional partitions of $n+k$ which contain $k$.
\end{itemize}
Thus, we have the number of partitions of $n+k$ containing $k$ as 
$R_k(n+k) = |Y| = P(n)$. The analysis holds for all $n, k > 0$, and hence the 
result.
\end{proof}

\subsection{The Main Result}
Now that we have proved the two lemmas required for our result, we can state and prove the main theorem of this paper, as follows.

\begin{theorem}[Extension of Stanley's Theorem]
\label{thm2}
Given any positive integer $n$ and any positive integer $k$,
$$S(n) \: = \: Q_k(n) + Q_k(n+1) + Q_k(n+2) + \cdots + Q_k(n+k-1) \: = \: 
\sum_{i=0}^{k-1} Q_k(n+i).$$
\end{theorem}
\begin{proof}
We shall prove this result by induction over $k$, as follows.

{\flushleft \sf Base Case ($k=1$):} $S(n) = Q_1(n)$ is true for all $n$ by 
Theorem~\ref{thm1} (Stanley's theorem).

{\flushleft \sf Assumption for $k$:} Let us assume that for some $k \geq 1$, 
and for all $n > 0$, $$S(n) = \sum_{i=0}^{k-1} Q_k(n+i).$$

{\flushleft \sf Proof for $k+1$:} Here, we have to prove that 
$S(n) = \sum_{i=0}^{k} Q_{k+1}(n+i)$ for all $n > 0$. Let us prove this result 
by induction on $n$, as follows.
\begin{description}
\item [$n=1$ :] For this base case of induction, we have 
$$\sum_{i=0}^{k} Q_{k+1}(1+i) = \sum_{i=0}^{k-1} Q_{k+1}(1+i) + Q_{k+1}(k+1) = 
0 + 1 = 1 = Q_1(1) = S(1).$$
\item [$n$ :] Let us assume that for some $n \geq 1$, $S(n) = Q_1(n) 
= \sum_{i=0}^{k} Q_{k+1}(n+i).$
\item [$n+1$ :] For this case with $n+1$, we have
\begin{eqnarray*}
S(n+1) & = & Q_1(n+1) \\
& = & Q_1(n) + R_1(n+1), \quad \mbox{from Lemma~\ref{lem1}}\\
& = & Q_1(n) + P(n), \quad \mbox{from Lemma~\ref{lem2}}\\
& = & \sum_{i=0}^{k} Q_{k+1}(n+i) + R_{k+1}(n+k+1), \quad \mbox{from assumption and Lemma~\ref{lem2}}\\
& = & \sum_{i=0}^{k} Q_{k+1}(n+i) + \left[Q_{k+1}(n+k+1) - Q_{k+1}(n)\right], \quad \mbox{from Lemma~\ref{lem1}}\\
& = & \sum_{i=1}^{k+1} Q_{k+1}(n+i) = \sum_{i=0}^{k} Q_{k+1}(n+i+1) = \sum_{i=0}^{k} Q_{k+1}((n+1)+i)
\end{eqnarray*}
\end{description}
Thus, by induction on $n$, the result is proved in case of $k+1$ for all 
$n>0$, and hence by induction on $k$, the result is proved for all $k > 0$ as 
well.
\end{proof}

\subsection{Illustrative Example}
In Table~\ref{tab2}, we present a modified version of the initial example ($n=4$) as an evidence to this observation and the validity of Theorem~\ref{thm2}. Recall from Table~\ref{tab1} that $S(4) = 7$. Each cell in the main segment of Table~\ref{tab2} represents the value of $Q_k(n)$ corresponding to the row $n$ and column $k$. The cells marked (-) are the values of $Q_k(n)$ that we do not require in Theorem~\ref{thm2}, and hence have been omitted from the table.

\begin{table}[htb]
\centering
\begin{tabular}{|c|c|c|c|c|}
\hline
\multicolumn{5}{|c|}{Theorem~\ref{thm2} for $n=4$}\\
\hline
\hline
$n \: \downarrow \quad k \: \rightarrow$ & 1 & 2 & 3 & 4 \\
\hline
4 & $Q_1(4) =$ 7 & $Q_2(4) =$ 3 & $Q_3(4) =$ 1 & $Q_4(4) =$ 1\\
\hline
5 & - & $Q_2(5) =$ 4 & $Q_3(5) =$ 2 & $Q_4(5) =$ 1 \\
\hline
6 & - & - & $Q_3(6) =$ 4 & $Q_4(6) =$ 2 \\
\hline
7 & - & - & - & $Q_4(7) =$ 3 \\
\hline
\hline
Total & $S(4) \: =$ {\bf 7} & $S(4) \: =$ {\bf 7} & $S(4) \: =$ {\bf 7} & $S(4) \: =$ {\bf 7}\\
\hline
\end{tabular}
\caption{Example of our Extension of Stanley's Theorem}
\label{tab2}
\end{table}

\section{Corollaries of the Extension}
\label{alter}
In this section we propose a few corollaries which follow directly from the Lemmas~\ref{lem1} and~\ref{lem2}, or from Theorem~\ref{thm2} itself. Our results can be exploited to provide alternate proofs of existing facts and analogues to established results in the field of integer partitions.

\subsection{Alternate Proofs of Existing Results}
One may find the following two results existing in the current 
literature~\cite{sloane1,sloane2} on partitions of integers. We shall provide alternate proofs to these results using Lemmas~\ref{lem1} and~\ref{lem2}.

\begin{result}
\label{res1}
Given any positive integer $n$, abiding by the previous notation, one has
$$ Q_1(n) = \sum_{i=0}^{n-1} P(i).$$
\end{result}
\begin{proof}
From Lemma~\ref{lem1}, we know that $Q_k(n+k) = Q_k(n) + R_k(n+k)$ for any pair of positive integers $(n, k)$. Also, from Lemma~\ref{lem2}, we get that $R_k(n+k) = P(n)$ for all pairs $(n, k)$. Choosing $k = 1$ and using these two lemmas appropriately, we obtain the following.
$$Q_1(n) = Q_1( (n-1) + 1) = Q_1(n-1) + R_1(n) = Q_1(n-1) + P(n-1)$$
Following an iterative model based on this relation, one may easily obtain the result as
$$Q_1(n) = \sum_{i = 1}^{n} P(i-1) = \sum_{i = 0}^{n-1} P(i).$$
Hence the result.
\end{proof}

\begin{result}
\label{res2}
Given any positive integer $n$, abiding by the previous notation, one has
$$Q_k(n) = 
\sum_{\begin{subarray}{c} i=0\\ i \equiv n \bmod k \end{subarray}}^{n-1} P(i).$$
\end{result}
\begin{proof}
Once again, we use the Lemmas~\ref{lem1} and~\ref{lem2} to get
$$Q_k(n) = Q_k( (n-k) + k) = Q_k(n-k) + R_k(n) = Q_k(n-k) + P(n-k).$$
An iterative process applied on this relation produces $Q_k(n) = \sum_{j > 0} P(n - jk)$. Recalling the fact that $P(n) = 0$ for $n < 0$, we obtain
$$Q_k(n) = \sum_{j > 0} P(n - jk) = \sum_{\begin{subarray}{c} i=0\\ i = n - jk, \: j \in \mathbb{Z} \end{subarray}}^{n-1} P(i) = \sum_{\begin{subarray}{c} i=0\\ i \equiv n \bmod k \end{subarray}}^{n-1} P(i).$$ 
Hence the result.
\end{proof}

\subsection{Alternate Proof of Theorem~\ref{thm2}}
Here we provide an alternate proof of Theorem~\ref{thm2}, using the two results proved in the previous subsection, instead of the lemmas that we used in the original proof.

\begin{theorem}[Theorem~\ref{thm2} Restated]
Given any pair of positive integers $(n, k)$,
$$S(n) \: = \: Q_k(n) + Q_k(n+1) + Q_k(n+2) + \cdots + Q_k(n+k-1) 
\: = \: \sum_{i=0}^{k-1} Q_k(n+i).$$
\end{theorem}
\begin{proof}
We have $S(n) = Q_1(n) = \sum_{i=0}^{n-1} P(i)$ by combining Stanley's theorem 
and Result~\ref{res1} as stated above. Let us define the following set of 
partitions.
$$P_n = \{ P(0), P(1), P(2), \ldots, P(n-1) \}.$$
Then, $S(n)$ is equal to the sum over all these partitions in $P_n$. Given the 
positive integer $k$, we can distribute $P_n$ over some disjoint copies of 
congruence classes as follows.
\begin{eqnarray*}
P_n & = & \{ P(i), P(i+1), \ldots, P(i+k-1) \: | \: i \equiv n \bmod k \} \\
& = & \bigcup_{j=0}^{k-1} \{ P(i+j) \: | \: 0 \leq i+j < n, \: i 
\equiv n \bmod k \} 
\end{eqnarray*} 
From the above-mentioned distribution of $P_n$, one may easily deduce that
\begin{eqnarray*}
S(n) \: = \: Q_1(n) & = & \sum_{i=0}^{n-1} P(i) \\
& = & \sum_{j=0}^{k-1} \sum_{\begin{subarray}{c} i+j=0\\ i 
\equiv n \bmod k \end{subarray}}^{n-1} P(i + j) \\ 
& = & \sum_{j=0}^{k-1} \sum_{\begin{subarray}{c} i+j=0\\ i+j 
\equiv n+j \bmod k \end{subarray}}^{n-1} P(i + j) \\
& = & \sum_{j=0}^{k-1} Q_k(n+j), \quad \mbox{by Result~\ref{res2}}.
\end{eqnarray*}
Hence the result, which holds true for any positive integral values of $n$ and $k$.
\end{proof}

Though the results proved in the previous subsection existed in the literature, there exists no attempt to derive this generalization of Stanley's theorem using these. This is to the best of our knowledge about the literature on integer partitions.

\subsection{Analogue of Ramanujan's Results}
In the theory of integer partitions, an array of elegant congruence relations for partition function $P(n)$ were proposed by Ramanujan. He proved that for every non-negative $n \in \mathbb{Z}$,
\begin{eqnarray*}
p(5n + 4) & \equiv & 0 \pmod{5},\\
p(7n + 5) & \equiv & 0 \pmod{7},\\
p(11n + 6) & \equiv & 0 \pmod{11},
\end{eqnarray*}
and he conjectured that there exist such congruence modulo arbitrary powers of 5, 7, 11. A lot of eminent mathematicians have worked on similar results for a long time, and the best result till date is: ``there exist such congruence relations for all non-negative integers which are co-prime to 6''. This result was proved by Ahlgren and Ono~\cite{ono1}. In this light, we propose a simple analogue to the Ramanujan results that follow directly from Lemmas~\ref{lem1} and~\ref{lem2}. 

In simple words, our first result can be stated as {\em ``for any non-negative integer $n$, the number of times 5 occurs in the unordered partitions of $5n + 4$ is divisible by 5''}. The other two results fall in the same line with divisibility by 7 and 11 respectively. Stated formally, the analogues of Ramanujan results that we propose are for the function $Q_k(n)$, as follows.

\begin{theorem}
Given any non-negative integer $n$, following the notation as before, one has
\begin{eqnarray*}
Q_{5} (5n + 4) & \equiv & 0 \pmod{5},\\
Q_{7} (7n + 5) & \equiv & 0 \pmod{7},\\
Q_{11} (11n + 6) & \equiv & 0 \pmod{11}.
\end{eqnarray*}
\end{theorem}
\begin{proof}
Let us prove the case for 5, and the other results will follow from similar idea. We will prove the result for 5 using an induction on $n$, as follows.
\begin{description}
\item [$n = 0$ :] In this case, $Q_5(4) = 0 \equiv 0 \pmod{5}$ is trivially true.

\item [$n = 1$ :] One can verify that $Q_5(9) = 5 \equiv 0 \pmod{5}$.

\item [$n$ :] Let us assume that $Q_5(5n + 4) \equiv 0 \pmod{5}$ for some integer $n > 1$.

\item [$n+1$ :] From Lemma~\ref{lem1} and~\ref{lem2}, using $k=5$, we obtain the following.
\begin{eqnarray*}
Q_5(5(n+1) + 4) = Q_5((5n+4) + 5) &=& Q_5(5n+4) + R_5((5n+4) + 5)\\ &=& Q_5(5n+4) + P(5n+4) \equiv 0 \pmod{5}
\end{eqnarray*}
The final congruence holds from our assumption and the first Ramanujan congruence.
\end{description}
Hence, by induction on $n$, the first result for modulus 5 is proved. One can similarly prove the results for moduli 7 and 11 using Lemmas~\ref{lem1} and~\ref{lem2} appropriately.
\end{proof}

One can also derive analogous results for higher order Ramanujan congruence, such as $Q_5(25n+24) \equiv 0 \pmod{5^2}$, $Q_5(125n+99) \equiv 0 \pmod{5^3}$ etc., in a similar fashion.

\subsection{Alternate Proof of Ramanujan's Results}
In this section, we propose a novel alternate proof of Ramanujan's congruence results for the partition function $P(n)$. To the best of our knowledge, this approach has not been proposed in the literature till date. For this alternate proof, we shall use the following relation, which we obtained in the previous section.
$$ P(5n+4) = Q_k(5(n+1) + 4) - Q_k(5n+4) $$
And in turn, we will require the study of $Q_k(n)$ in terms of generating functions to prove Ramanujan's results for $P(n)$. Let us prove the first relation (case with modulo 5), and the other proofs will follow similarly. 

\subsubsection{Generating Function of $Q_k(n)$}
The generating function for the partition function $P(n)$ is given by
$$ F(x) = \sum_{m=0}^{\infty} P(m) \cdot x^m  = \prod_{n=1}^{\infty} \frac{1}{1 - x^n}$$ 
where we assume $P(0) = 1$. In this formula, we count the coefficient of $x^m$ on both sides, where the coefficient on the right hand side is the result of counting all possible ways that $x^m$ is generated by multiplying smaller or equal powers of $x$. This obviously gives the number of partitions of $m$ into smaller or equal parts.

What we require for $Q_k(n)$ is to count the number of $k$'s occurring in each of these partitions. Thus, we want to (i) add $r$ to the count if $x^{rk}$ is a member involved from the right hand side, and (ii) not count any of the partitions where no power of $x^k$ is involved. This intuition gives rise to the following generating function for $Q_k(n)$.
\begin{eqnarray*}
G_k(x) = \sum_{m=0}^{\infty} Q_k(m) \cdot x^m  &=& \frac{1 \cdot x^k + 2 \cdot x^{2k} + 3 \cdot x^{3k} + \cdots }{(1 - x) \cdot (1 - x^2) \cdots (1 - x^{k-1}) \cdot (1 - x^{k+1}) \cdots }\\
& = & (1 - x^k) \cdot (x^k + 2 x^{2k} + 3 x^{3k} + \cdots )  \cdot \prod_{n=1}^{\infty} \frac{1}{1 - x^n} \\
& = & (x^k + x^{2k} + x^{3k} + \cdots )  \cdot \prod_{n=1}^{\infty} \frac{1}{1 - x^n} \\
& = & \frac{x^k}{1 - x^k} \cdot \prod_{n=1}^{\infty} \frac{1}{1 - x^n}.
\end{eqnarray*}
Our next goal is to use the generating function $G_5(x)$ to prove $Q_5(5n+4) \equiv 0 \pmod{5}$.

\subsubsection{Proof of $Q_5(5n+4) \equiv 0 \pmod{5}$}
We prove this along the same line as Ramanujan's proof for $P(5n+4)$. Notice that we have
\begin{align}
\label{rameq1}
& x \cdot [(1-x)\cdot(1-x^2)\cdot(1-x^3)\cdots]^4 \nonumber \\
& \quad \quad = x \cdot (1 - 3x + 5x^2 - 7x^3 + \cdots) \cdot (1 - x - x^2 + x^3 + \cdots) \nonumber \\
& \quad \quad = \sum_{\mu = 0}^{\infty} \sum_{\nu = -\infty}^{\infty} (-1)^{\mu + \nu} \cdot (2\mu + 1) \cdot x^{1+\frac{1}{2}\mu(\mu+1)+\frac{1}{2}\nu(3\nu+1)}
\end{align}
In the above expression, let us consider the coefficient of $x^{5n}$. In this case, we have
\begin{align}
\label{rameq2}
& 1+\frac{1}{2}\mu(\mu+1)+\frac{1}{2}\nu(3\nu+1) \equiv 0 \pmod{5} \nonumber\\
\Rightarrow \quad & 8+ 4\mu(\mu+1)+4\nu(3\nu+1) \equiv 0 \pmod{5} \nonumber\\
\Rightarrow \quad & (2\mu+1)^2 + 2(\nu+1)^2 \equiv 0 \pmod{5}
\end{align}
Note that $(2\mu+1)^2$ is congruent to 0, 1 or 4 modulo 5, whereas $2(\nu+1)^2$ is congruent to 0, 2 or 3 modulo 5. Thus, to satisfy Equation~\eqref{rameq2}, we must have both $(2\mu+1)^2$ and $2(\nu+1)^2$ congruent to 0 modulo 5, i.e., $(2\mu + 1) \equiv 0 \pmod{5}$ to be specific. Using this in conjunction with Equation~\eqref{rameq1}, we get that the coefficient of $x^{5n}$ in $x \cdot [(1-x)\cdot(1-x^2)\cdot(1-x^3)\cdots]^4$ is divisible by 5.

Again, considering the coefficients of $\frac{1}{1-x^5}$ and $\frac{1}{(1-x)^5}$ modulo 5, we get
$$ \frac{1}{(1-x)^5} \equiv \frac{1}{1 - x^5} \pmod{5} \quad \Rightarrow \quad \frac{1 - x^5}{(1-x)^5} \equiv 1 \pmod{5}. $$
Thus, all the coefficients (except the first one) in the expression
$$ \frac{1 - x^5}{(1-x)^5} \cdot \frac{1 - x^{10}}{(1-x^2)^5} \cdot \frac{1 - x^{15}}{(1-x^3)^5} \cdots = [(1-x^5)\cdot(1-x^{10})\cdots ] \cdot \prod_{n=1}^{\infty} \frac{1}{(1-x^n)^5} $$
are divisible by 5. Hence, the coefficient of $x^{5n}$ in the expansion of
\begin{align*}
& x\cdot[(1-x^5)\cdot(1-x^{10})\cdots ] \cdot \prod_{n=1}^{\infty} \frac{1}{1-x^n}\\ 
& \: = x\cdot[(1-x)\cdot(1-x^2)\cdot(1-x^3)\cdots ]^4 \cdot [(1-x^5)\cdot(1-x^{10})\cdots ] \cdot \prod_{n=1}^{\infty} \frac{1}{(1-x^n)^5}
\end{align*}
is also divisible by 5, by the virtue of the previous discussion. This in turn proves that the coefficient of $x^{5n}$ is the expansion of
$$ x \cdot \frac{x^5}{1-x^5} \cdot \prod_{n=1}^{\infty} \frac{1}{1-x^n} = x^5 \cdot (1+x^5+x^{10}+x^{15}+\cdots) \cdot x \cdot \prod_{n=1}^{\infty} \frac{1}{1-x^n} $$
is divisible by 5. Now, recall the generating function of $Q_5(n)$:
$$ G_5(x) = \sum_{m=0}^{\infty} Q_5(m) \cdot x^m = \frac{x^5}{1-x^5} \cdot \prod_{n=1}^{\infty} \frac{1}{1-x^n}. $$
The discussion so far reveals that the coefficient of $x^{5n}$ in the expansion of $x \cdot G_5(x)$ is divisible by 5, i.e., the coefficient of $x^{5n-1}$ or $x^{5n+4}$ in the expansion of $G_5(x)$ is divisible by 5. Therefore, we have the desired result:
$$ Q_5(5n+4) \equiv 0 \pmod{5} \quad \mbox{ for all non-negative $n \in \mathbb{Z}$.} $$

\subsubsection{Proof of $P(5n+4) \equiv 0 \pmod{5}$}
Recall the following result from the proof of the analogue of Ramanujan's results.
$$ P(5n+4) = Q_5(5(n+1)+4) - Q_5(5n+4) \quad \mbox{ for all non-negative $n \in \mathbb{Z}$.} $$
This readily provides us with $P(5n + 4) \equiv 0 \pmod{5}$ for all non-negative $n\in \mathbb{Z}$. This gives a new alternate proof of Ramanujan's congruence results using the concept of $Q_k(n)$. The results for $k = 7, 11$ may be derived similarly.

\section{Conclusion}
In this paper we have presented an extension of Stanley's theorem for partitions of positive integers. There exists an extension of Stanley's theorem in the literature, which had been proposed by Elder in 1984. But our generalization opens a new line of thought in this direction. We have also proposed some interesting corollaries to illustrate potential applications of this generalization towards the existing results in integer partitions. 

In particular, we prove a couple of existing results using alternative techniques and propose analogues to Ramanujan's results on congruence behavior of functions related to partitions. We have also proposed an alterate proof of Ramanujan's results using the concept of $Q_k(n)$ and its generating function.

A couple of interesting research problems related to this topic may be to study Elder's theorem in the light of the results proposed in this paper, and to propose analogues to more involved congruence results for $P(n)$ proved in the literature.

{\flushleft \bf Acknowledgement:} The authors would like to thank their advisor Professor Subhamoy Maitra of Applied Statistics Unit, Indian Statistical Institute, for numerous technical discussions and untiring support throughout the project. The authors also like to express their gratitude towards Professor Pantelimon Stanica for his preliminary feedback on the technical content of this paper.

\end{document}